\documentclass[pra,superscriptaddress,amsmath,amssymb,aps,twosided,twocolumn,nopacs]{revtex4-1}

\usepackage{bm}
\usepackage{graphicx,epic,eepic,epsfig,amsmath,latexsym,amssymb,verbatim,color}
\usepackage{theorem}
\usepackage[
breaklinks=true,colorlinks=true,
linkcolor=green,urlcolor=green,citecolor=green,% PDF VIEW
linkcolor=blue,urlcolor=blue,citecolor=blue,% PRINT
bookmarks=true,bookmarksopenlevel=2]{hyperref}
\newtheorem{definition}{Definition}
\newtheorem{proposition}[definition]{Proposition}
\newtheorem{lemma}[definition]{Lemma}

\newtheorem{theorem}[definition]{Theorem}

\def\squareforqed{\hbox{\rlap{$\sqcap$}$\sqcup$}}
\def\qed{\ifmmode\squareforqed\else{\unskip\nobreak\hfil
\penalty50\hskip1em\null\nobreak\hfil\squareforqed
\parfillskip=0pt\finalhyphendemerits=0\endgraf}\fi}
\def\endenv{\ifmmode\;\else{\unskip\nobreak\hfil
\penalty50\hskip1em\null\nobreak\hfil\;
\parfillskip=0pt\finalhyphendemerits=0\endgraf}\fi}
\newenvironment{proof}{\noindent \textbf{{Proof~} }}{\qed}
\newenvironment{remark}{\noindent \textbf{{Remark~}}}{}

\mathchardef\ordinarycolon\mathcode`\:
\mathcode`\:=\string"8000
\def\vcentcolon{\mathrel{\mathop\ordinarycolon}}
\begingroup \catcode`\:=\active
  \lowercase{\endgroup
  \let :\vcentcolon
  }

\newcommand{\nc}{\newcommand}
\nc{\rnc}{\renewcommand}
\nc{\beq}{\begin{equation}}
\nc{\eeq}{\end{equation}}
\nc{\beqa}{\begin{eqnarray}}
\nc{\eeqa}{\end{eqnarray}}
\nc{\lbar}[1]{\overline{#1}}
\nc{\bra}[1]{\langle#1|}
\nc{\ket}[1]{|#1\rangle}
\nc{\ketbra}[2]{|#1\rangle\!\langle#2|}
\nc{\braket}[2]{\langle#1|#2\rangle}
\nc{\proj}[1]{| #1\rangle\!\langle #1 |}
\nc{\avg}[1]{\langle#1\rangle}
\nc{\Rank}{\operatorname{Rank}}
\nc{\smfrac}[2]{\mbox{$\frac{#1}{#2}$}}
\nc{\tr}{\operatorname{Tr}}
\nc{\ox}{\otimes}
\nc{\dg}{\dagger}
\nc{\dn}{\downarrow}
\nc{\cA}{{\cal A}}
\nc{\cB}{{\cal B}}
\nc{\cC}{{\cal C}}
\nc{\cD}{{\cal D}}
\nc{\cE}{{\cal E}}
\nc{\cF}{{\cal F}}
\nc{\cG}{{\cal G}}
\nc{\cH}{{\cal H}}
\nc{\cI}{{\cal I}}
\nc{\cJ}{{\cal J}}
\nc{\cK}{{\cal K}}
\nc{\cL}{{\cal L}}
\nc{\cM}{{\cal M}}
\nc{\cN}{{\cal N}}
\nc{\cO}{{\cal O}}
\nc{\cP}{{\cal P}}
\nc{\cQ}{{\cal Q}}
\nc{\cR}{{\cal R}}
\nc{\cS}{{\cal S}}
\nc{\cT}{{\cal T}}
\nc{\cX}{{\cal X}}
\nc{\cY}{{\cal Y}}
\nc{\cZ}{{\cal Z}}
\nc{\csupp}{{\operatorname{csupp}}}
\nc{\qsupp}{{\operatorname{qsupp}}}
\nc{\diag}{{\operatorname{diag}}}
\nc{\var}{{\operatorname{var}}}
\nc{\rar}{\rightarrow}
\nc{\lrar}{\longrightarrow}
\nc{\polylog}{{\operatorname{polylog}}}
\nc{\1}{{\openone}}
\nc{\wt}{{\operatorname{wt}}}
\nc{\av}[1]{{\left\langle {#1} \right\rangle}}

\nc{\RR}{{{\mathbb R}}}
\nc{\CC}{{{\mathbb C}}}
\nc{\FF}{{{\mathbb F}}}
\nc{\NN}{{{\mathbb N}}}
\nc{\ZZ}{{{\mathbb Z}}}
\nc{\PP}{{{\mathbb P}}}
\nc{\QQ}{{{\mathbb Q}}}
\nc{\UU}{{{\mathbb U}}}
\nc{\EE}{{{\mathbb E}}}
\nc{\id}{{\operatorname{id}}}

\nc{\CHSH}{{\operatorname{CHSH}}}

\nc{\be}{\begin{equation}}
\nc{\ee}{{\end{equation}}}
\nc{\bea}{\begin{eqnarray}}
\nc{\eea}{\end{eqnarray}}
\nc{\<}{\langle}
\rnc{\>}{\rangle}
\nc{\Hom}[2]{\mbox{Hom}(\CC^{#1},\CC^{#2})}
\nc{\rU}{\mbox{U}}

\nc{\ob}[1]{#1}

\begin{document}

\title{Resource Theory of Coherence --- Beyond States}

\author{Khaled Ben Dana}%
\email{khaled.bendana@gmail.com}
\affiliation{Laboratory of Advanced Materials and Quantum Phenomena, Department of Physics, 
             Faculty of Sciences of Tunis, University Tunis El Manar, Tunis 2092, Tunisia}

\author{Mar\'ia Garc\'ia D\'iaz}
\email{maria.garciadia@e-campus.uab.cat}
\affiliation{Departament de F\'isica: Grup d'Informaci\'o Qu\`antica,
             Universitat Aut\`onoma de Barcelona, ES-08193 Bellaterra (Barcelona), Spain}

\author{Mohamed Mejatty}%
\email{mohamedmejatty@yahoo.fr}
\affiliation{Laboratory of Advanced Materials and Quantum Phenomena, Department of Physics, 
             Faculty of Sciences of Tunis, University Tunis El Manar, Tunis 2092, Tunisia}

\author{Andreas Winter}%
\email{andreas.winter@uab.cat}
\affiliation{Departament de F\'isica: Grup d'Informaci\'o Qu\`antica,
             Universitat Aut\`onoma de Barcelona, ES-08193 Bellaterra (Barcelona), Spain}
\affiliation{ICREA---Instituci\'o Catalana de Recerca i Estudis Avan\c{c}ats, 
             Pg. Llu\'is Companys, 23, ES-08001 Barcelona, Spain}

\date{22 May 2017}

\begin{abstract}
We generalize the recently proposed resource theory of coherence (or superposition) 
[Baumgratz, Cramer \&{} Plenio, Phys. Rev. Lett. 113:140401; 
Winter \&{} Yang, Phys. Rev. Lett. 116:120404] 
to the setting where not only the free (``incoherent'') resources,
but also the objects manipulated, are quantum operations rather than states. 

In particular, we discuss an information theoretic notion of
coherence capacity of a quantum channel, and prove a single-letter
formula for it in the case of unitaries. 
Then we move to the coherence cost of simulating a 
channel, and prove achievability results for
unitaries and general channels acting on a $d$-dimensional system;
we show that a maximally coherent state of rank $d$ is always
sufficient as a resource if incoherent operations are allowed,
and rank $d^2$ for ``strictly incoherent'' operations.
We also show lower bounds on the simulation cost of channels
that allow us to conclude that there exists bound coherence
in operations, i.e.~maps with non-zero cost of implementing
them but zero coherence capacity; this is in contrast to states, which
do not exhibit bound coherence.
\end{abstract}

\maketitle

\section{Introduction}
\label{sec:intro}
Since its discovery to our days, quantum mechanics has provided a mathematical framework 
for the construction of physical theories. %\cite{Nielsen}.
Indeterminism, interference, uncertainty, superposition and entanglement are concepts of quantum 
mechanics that distinguish it from classical physics, and which have become resources
in quantum information processing. 
%This means that they help us achieve tasks that are not possible in classical physics. 

Quantum resource theories aim at capturing the essence of these traits 
and quantifying them. Recently, quantum resource theories have been 
formulated in different areas of physics such as the resource theory of athermality 
in thermodynamics \cite{Brandao1,Brandao2,Horodecki,Faist,Gour,Narasimhachar}
and the resource theory of asymmetry \cite{Gour1,Marvian}. 
%and the resource theory of non-Gaussianity in quantum optics \cite{Browne}. 
Furthermore, general structural frameworks of quantum 
resource theories have been proposed \cite{Brandao3}.

Resource theories using concepts of quantum mechanics have been developed since 
the appearance of the theory of entanglement \cite{Plenio1,Vedral,HHHH}. 
Very recently, Baumgratz \emph{et al.} \cite{Baumgratz}, 
following earlier work by \AA{}berg \cite{Aaberg}, have made quantum 
coherence itself, i.e.~the concept of superposition
\[
  \ket{\psi} = \sum_i \psi_i \ket{i},
\]
the subject of a resource theory; see also \cite{Plenio-2}.

The present paper is concerned with this resource theory of coherence,
and here we briefly recall its fundamental definitions, as well as some 
important coherence measures; for a comprehensive review, see \cite{Streltsov-rev}. 
Let $\{\ket{i} : i = 0,\ldots,d-1\}$ 
be a particular fixed basis of the $d$-dimensional Hilbert space $\cH$;
then all density matrices in this basis are ``incoherent'', 
i.e.~those of the form 
$\tilde{\delta}=\sum_{i=0}^{d-1} \delta_{i}\proj{i}$. We denote by
$\Delta\subset \cS(\cH)$ the set of such incoherent quantum states. 

The definition of coherence monotones requires the identification of 
operations that are incoherent. These map the set of incoherent states 
to itself. More precisely, such a completely positive and trace preserving 
(cptp) map is specified by a set of Kraus operators $\{K_{\alpha}\}$ 
satisfying $\sum_{\alpha} K_{\alpha}^{\dagger}K_{\alpha}=\1$ and 
$K_{\alpha}{\Delta}K_{\alpha}^{\dagger}\subset{\Delta}$ for all $\alpha$.
A Kraus operator with this property is called \emph{incoherent};
we call it \emph{strictly incoherent} if both $K$ and $K^\dagger$ are
incoherent \cite{Winter,Yadin}.
We distinguish two classes of incoherent operations (IO): 

(i) Incoherent completely positive and trace 
preserving quantum operations (non-selective maps) $T$, 
which act as $T(\rho)=\sum_{\alpha}K_{\alpha}\rho K_{\alpha}^{\dagger}$
(note that this formulation implies the loss of information about 
the measurement outcome);

(ii) quantum operations for which measurement outcomes are retained, 
given by $\rho_{\alpha}=\frac{1}{p_\alpha}K_{\alpha}\rho K_{\alpha}^{\dagger}$ occurring
with probability $p_{\alpha} = \tr K_{\alpha}\rho K_{\alpha}^{\dagger}$. 
The latter can be modelled as a nonselective operation by explicitly 
introducing a new register to hold the (incoherent) measurement result:
\[
  \widetilde{T}(\rho) = \sum_\alpha K_\alpha \rho K_\alpha^\dagger \otimes \proj{\alpha}.
\]
Here, we have made use of the convention that when composing systems, the
incoherent states in the tensor product space are precisely the tensor products of
incoherent states and their probabilistic mixtures (convex combinations) \cite{Aaberg}.
An operation is called strictly incoherent (SIO), if it can be written
with strictly incoherent Kraus operators.

We define also the maximally coherent state on a $d$-dimensional system by 
$\ket{\Psi_{d}} = \frac{1}{\sqrt{d}}\sum_{i=0}^{d-1} \ket{i}$, from which every state 
in dimension $d$ can be prepared \cite{Baumgratz,Winter,Du1}. 
Note that the definition of maximally coherent state is 
independent of a specific measure for the coherence \cite{Baumgratz,Du2}.

A list of desirable conditions for any coherence measure, i.e.~a functional
from states to non-negative real numbers, was also presented \cite{Baumgratz}:
\begin{enumerate}
  \item $C(\rho)=0$ for all $\rho\in \Delta$;
  \item Monotonicity under non-selective incoherent maps: $C(\rho) \geq C(T(\rho))$;
  \item (Strong) monotonicity under selective incoherent maps:
        $C(\rho) \geq \sum_{\alpha} p_\alpha C(\rho_{\alpha})$;
  \item Convexity: $\sum_i p_i C(\rho_i) \geq C(\sum_i p_i\rho_i)$.
\end{enumerate} 
The first two are definitely required to speak of a coherence measure, 
the third is sometimes demanded axiomatically, but often is really more of a 
convenience, and convexity should be thought of as nice if present but
not absolutely necessary.

Among the most important examples are the following three measures:
For pure states $\varphi = \proj{\varphi}$, the \emph{entropy of coherence}
is defined as
\begin{equation}
  C(\varphi) := S\bigl(\Delta(\varphi)\bigr),
\end{equation}
where $\Delta$ is the dephasing (i.e. coherence-destroying) map
$\Delta(\rho) = \sum_i \proj{i} \rho \proj{i}$; for mixed states, it
is extended by the convex hull construction to the 
\emph{coherence of formation}~\cite{Aaberg,Baumgratz}
\begin{equation}
  C_f(\rho) := \min \sum_i p_i C(\psi_i) \text{ s.t. } \rho = \sum_i p_i \psi_i.
\end{equation}
Finally, the \emph{relative entropy of coherence} \cite{Aaberg,Baumgratz}
\begin{equation}
  C_r(\rho) := \min_{\sigma\in\Delta} D(\rho\|\sigma) = S\bigl(\Delta(\rho)\bigr)-S(\rho),
\end{equation}
with the relative entropy $D(\rho\|\sigma)=\tr\rho(\log\rho-\log\sigma)$.
Both $C_r$ and $C_f$ satisfy all properties 1 through 4 above.

\medskip
In the present paper, we expand our view from states as coherent resources
to operations, showing how to extract pure state coherence from a 
given operation (section~\ref{sec:power}), 
and how to implement operations using coherent states as a resource
(section~\ref{sec:general}). We briefly discuss the case of qubit
unitaries as an example (section~\ref{sec:qubit}), and conclude
in section~\ref{sec:outlook}, where we observe that in operations, there
is bound coherence, something that doesn't exist for states.
Most generally, we propose a definition for the rate of conversion
between two channels using only incoherent operations, and present
many open questions about this concept.

%We show how to implement a $d$-dimensional unitary operation 
%from maximally coherent state by incoherent operations. Then we give a 
%solution for $2\times 2$ unitaries from suitable resource state. 
%Fulfilling the condition imposed by the majorization for transforming an 
%arbitrary pure state to another, we show how to implement a unitary from 
%another by means of incoherent operations. This result allow to us to interest 
%to the problem of turning $U$ into $V$ at the asymptotic limit. Using 
%the concepts of coherence distillation and coherence formation, defined 
%in \cite{Winter}, we show that the optimal rate for transforming $U$ 
%to $V$ is the relative entropy of coherence of $U \proj{0} U^{\dagger}$. 

\section{Coherence generating capacity and coherence power of transformations}
\label{sec:power}
The free operations in our resource theory are the incoherent ones (IO),
which means that in some sense, any other CPTP map represents a resource.
How to measure it? Or better: How to assess its resource character? 
In the present section, we are focusing on how much pure state 
coherence can be created asymptotically, using a given operation 
$T:A \longrightarrow B$ a large number of times, when incoherent 
operations are for free.

The most general protocol to generate coherence must use the resource
$T$ and incoherent operations according to some predetermined
algorithm, in some order. We may assume that the channels $T$ are 
invoked one at a time; and we can integrate all incoherent
operations in between one use of $T$ and the next into one
incoherent operation, since IO is closed under composition.
Thus, a mathematical description of the most general protocol
is the following: One starts by preparing an incoherent state $\rho_0$ on
$A\otimes A_0$, then lets act $T$, followed by 
an incoherent transformation $\mathcal{I}_1:B\otimes A_0 \longrightarrow A\otimes A_1$,
resulting in the state
\[
  \rho_1 = \mathcal{I}_1\bigl( (T\otimes\id)\rho_0 \bigr).
\]
Iterating, given the state $\rho_t$ on $A \otimes A_t$ obtained after the 
action of $t$ realizations of $T$ and suitable incoherent operations,
we let $T$ act and the incoherent transformation 
$\mathcal{I}_{t+1}:B\otimes A_t \longrightarrow A\otimes A_{t+1}$, resulting in the state
\[
  \rho_{t+1} = \mathcal{I}_{t+1}\bigl( (T\otimes\id)\rho_t \bigr).
\]
At the end of $n$ iterations, we have a state $\rho_n$ on $A\otimes A_n$,
and we call the above procedure a \emph{coherence generation protocol}
of rate $R$ and error $\epsilon$, if $|A_n| = 2^{nR}$ and the reduced
state $\rho_n^{A_n} = \tr_A \rho_n$ has high fidelity with the maximally
coherent state,
\[
  \bra{\Psi_{2^{nR}}} \rho_n^{A_n} \ket{\Psi_{2^{nR}}} \geq 1-\epsilon.
\]
The maximum number $R$ such that there exist coherence
generating protocols for all $n$, with error going to zero and
rates converging to $R$, is called the \emph{coherence generating capacity}
of $T$, and denoted $C_{\text{gen}}(T)$.
 
\begin{theorem}
  \label{thm:C-gen-T}
  For a general CPTP map $T:A \longrightarrow B$,
  \begin{equation}
    \label{eq:T-lower}
    C_{\text{gen}}(T) \geq \sup_{\ket{\varphi}\in A\otimes C} C_r\bigl((T\otimes\id)\varphi\bigr) - C(\varphi),
  \end{equation}
  where the supremum over all auxiliary systems $C$ and pure states 
  $\ket{\varphi}\in A\otimes C$. Furthermore,
  \begin{equation}
    \label{eq:T-upper}
    C_{\text{gen}}(T) \leq \sup_{\rho \text{ on } A\otimes C} C_r\bigl((T\otimes\id)\rho\bigr) - C_r(\rho),
  \end{equation}
  where now the supremum is over mixed states $\rho$ on $A\otimes C$.
  
  If $T$ is an isometry, i.e.~$T(\rho) = V\rho V^\dagger$ for an
  isometry $V:A\hookrightarrow B$, 
  the lower bound is an equality, and can be simplified:
  \begin{align}
    \label{eq:iso-formula}
    C_{\text{gen}}(V\!\cdot\! V^\dagger) 
        &= \sup_{\ket{\varphi}\in A\otimes C} 
                   C\bigl((V\otimes\1)\varphi(V\otimes\1)^\dagger\bigr) - C(\varphi) \nonumber \\
        &= \max_{\ket{\varphi}\in A} 
                   C\bigl( V\varphi V^\dagger\bigr) - C(\varphi).
  \end{align}
\end{theorem}

This result, the main one of the present section, should be compared to
the formula, similar in spirit, for the entangling power of a bipartite
unitary~\cite{Bennett,Collins}. 
Furthermore, the above formulas for the coherence generating capacity are 
related to the \emph{coherence power} (w.r.t.~ the relative entropy measure)
\begin{equation}
  P_r(T) = \max_{\rho \text{ on } A} C_r\bigl(T(\rho)\bigr) - C_r(\rho),
\end{equation}
investigated by Garc\'ia D\'iaz \emph{et al.}~\cite{Garcia-Diaz} 
and Bu \emph{et al.}~\cite{coherence-power}. 
Let us also introduce the same maximization restricted to pure input states,
\begin{equation}
  \widetilde{P}_r(T) = \max_{\ket{\varphi}\in A} C_r\bigl(T(\varphi)\bigr) - C(\varphi).
\end{equation}
Note that the only difference to our formulas is that
we allow an ancilla system $C$ of arbitrary dimension. If we consider,
for a general CPTP map $T$, the extension $T\otimes\id_k$ and
\begin{equation}
  P_r^{(k)}(T) := P_r(T \otimes \id_k), \quad
  \widetilde{P}_r^{(k)}(T) := \widetilde{P}_r(T \otimes \id_k),  
\end{equation}
then we have
\begin{equation}
  C_{\text{gen}}(V\!\cdot\! V^\dagger) = \sup_k \widetilde{P}_r^{(k)}(V\!\cdot\! V^\dagger)
                                       = \widetilde{P}_r(V\!\cdot\! V^\dagger),
\end{equation}
and in general
\begin{equation}
  \label{eq:cap-vs-power}
  \sup_k \widetilde{P}_r^{(k)}(T) \leq C_{\text{gen}}(T) \leq \sup_k P_r^{(k)}(T).
\end{equation}

\begin{proof}
We start with the lower bound, Eq.~(\ref{eq:T-lower}): For a given ancilla $C$ and
$\ket{\varphi} \in A \otimes C$, let $R = C_r\bigl((T\otimes\id)\varphi\bigr) - C(\varphi)$.
For any $\epsilon,\delta > 0$, we can choose, by the results of \cite{Winter},
a sufficiently large $n$ such that
\begin{align*}
  \Psi_2^{\otimes \lfloor n C(\varphi)+n\delta \rfloor}
                   &\stackrel{\text{IO}}{\longmapsto} \approx \varphi^{\otimes n}, \\
  \rho^{\otimes n} &\stackrel{\text{IO}}{\longmapsto} \approx \Psi_2^{\otimes \lceil nC_r(\rho)-n\delta \rceil},
\end{align*}
with $\rho = (T\otimes\id)\varphi$, and where $\approx$ refers to approximation
of the target state up to $\epsilon$ in trace norm. We only have to prove something
when $R > 0$, which can only arise if $T$ is not incoherent, meaning that there
exists an initial state $\ket{0}$ mapped to a coherent resource $\sigma = T(\proj{0})$,
i.e.~$C_r(\sigma) > 0$. In the following, assume $R > 2\delta$.
Now, we may assume that $n$ is large enough so that with 
$R_0=\frac{C(\varphi)+\delta}{C_r(\sigma)}+\delta$,
\[
  \sigma^{\otimes \lfloor n R_0 \rfloor} 
        \stackrel{\text{IO}}{\longmapsto} \approx \Psi_2^{\otimes \lfloor n C(\varphi)+n\delta \rfloor}.
\]
The protocol consists of the following steps:

\textbf{0}: Use $\lfloor n R_0 \rfloor$ instances of $T$ to create as many copies of $\sigma$,
 and convert them into $\Psi_2^{\otimes \lfloor n C(\varphi)+n\delta \rfloor}$ (up to
 trace norm $\epsilon$). 

\textbf{1-k (repeat)}: Convert first $\lfloor n C(\varphi)+n\delta \rfloor$
 of the already created copies of $\Psi_2$ into $n$ copies of $\varphi$; then apply $T$ to each 
 of them to obtain $\rho = (T\otimes\id)\varphi$; 
 and convert the $n$ copies of $\rho$ to $\Psi_2^{\otimes \lceil nC_r(\rho)-n\delta \rceil}$,
 incurring an error of $2\epsilon$ in trace norm in each repetition.

At the end, we have $(k-1)n (R-2\delta) + n C(\varphi)$ copies of $\Psi_2$, up to
trace distance $O(k^2\epsilon)$, using the channel a total of $kn + nR_0$ times, 
i.e.~the rate is $\geq (R-2\delta)\frac{k-1}{k+R_0}$, which can be made arbitrarily 
close to $R$ by choosing $\delta$ small enough and $k$ large enough (which in
turn can be effected by sufficiently small $\epsilon$).

For the upper bound, Eq.~(\ref{eq:T-upper}), consider a generic protocol
using the channel $n$ times, starting from $\rho_0$ (incoherent) and
generating $\rho_1, \ldots, \rho_n$ step by step along the way, such that
$\rho_n$ has fidelity $\geq 1-\epsilon$ with $\Psi_2^{\otimes nR}$. 
By the asymptotic continuity of $C_r$ \cite[Lemma~12]{Winter}, 
$C_r(\rho_n) \geq nR - 2\delta n - 2$, with $\delta = \sqrt{\epsilon(2-\epsilon)}$,
so we can bound
\[\begin{split}
  nR-2\delta n-2 &\leq C_r(\rho_n) \\
                 &=    \sum_{t=0}^{n-1} C_r(\rho_{t+1}) - C_r(\rho_t) \\
                 &\leq \sum_{t=1}^{n-1} C_r\bigl( (T\otimes\id)\rho_t \bigr) - C_r(\rho_t),
\end{split}\]
where we have used the fact that $\rho_0$ is incoherent and that
$\rho_{t+1} = \mathcal{I}_{t+1}\bigl( (T\otimes\id)\rho_t \bigr)$, with an incoherent
operation $\mathcal{I}_{t+1}$, which can only decrease relative entropy of coherence.
However, each term on the right hand sum is of the form
$C_r\bigl((T\otimes\id)\rho\bigr) - C_r(\rho)$ for a suitable ancilla $C$
and a state $\rho$ on $A\otimes C$. Thus, dividing by $n$ and letting 
$n\rightarrow\infty$, $\epsilon\rightarrow 0$ shows that 
$R \leq \sup_{\rho \text{ on } A\otimes C} C_r\bigl((T\otimes\id)\rho\bigr) - C_r(\rho)$.

For an isometric channel $T(\rho) = V\rho V^\dagger$, note that the initial state
$\rho_0$ in a general protocol is without loss of generality pure, and that $T$
maps pure states to pure states. The incoherent operations $\mathcal{I}_t$ map pure states
to \emph{ensembles} of pure states, so that following the same converse reasoning
as above, we end up upper bounding $R$ by an average of expressions
$C_r\bigl( (T\otimes\id)\rho \bigr) - C_r(\rho)$, with pure states $\rho$,
i.e.~Eq.~(\ref{eq:iso-formula}), since we also have
$C_{\text{gen}}(T) \geq C\bigl((T\otimes\id)\varphi\bigr) - C(\varphi)$
from the other direction.
The fact that no ancilla system is needed, is an elementary calculation.
Indeed, for a pure state $\ket{\varphi}\in A\otimes C$,
\[\begin{split}
  C\bigl((T\otimes\id)\varphi) &-C(\varphi) \\
     &= S\bigl((\Delta\circ T\otimes\Delta)\varphi)-S\bigl((\Delta\otimes\Delta)\varphi\bigr) \\
     &= S\bigl((\Delta\circ T\otimes\id)\rho)-S\bigl((\Delta\otimes\id)\rho\bigr), \\
     &\phantom{========}
      \bigl[\text{with } \rho = (\id\otimes\Delta)\varphi \\
     &\phantom{========[\text{with } \rho}  
                              = \sum_i p_i \varphi_i \otimes \proj{i} \bigr], \\
     &= \sum_i p_i \Bigl( S\bigl( \Delta(T(\varphi_i)) \bigr) - S\bigl( \Delta(\varphi_i) \bigr) \Bigr) \\
     &\leq \max_{\ket{\varphi}\in A} S\bigl( \Delta(T(\varphi_i)) \bigr) - S\bigl( \Delta(\varphi_i) \bigr),
\end{split}\]
and we are done.
\end{proof}

\medskip
\begin{remark}
We do not know, at this point, whether the suprema over the ancillary
systems in the upper and lower bounds in Eq.~(\ref{eq:cap-vs-power}) 
are necessary in general, i.e.~it might be that 
$P_r(T) = P_r^{(k)}$ and/or $\widetilde{P}_r(T) = \widetilde{P}_r^{(k)}$;
note that the latter is the case for unitaries, even though
it seems unlikely in general, cf.~\cite{Bennett,Collins}. We do know,
however, that $P_r$ is convex and non-increasing under composition
of the channel with incoherent operations \cite[Cor.~1 and ~2]{Garcia-Diaz}.

It should be appreciated that even the calculation of $P_r(T)$
appears to be a hard problem. The investigation of further questions, such 
as the additivity of $P_r$, $\widetilde{P}_r$ or $C_{\text{gen}}$, 
may depend on making progress on that problem.
\end{remark}

\medskip
\begin{remark}
The same reasoning as in the proof of Theorem~\ref{thm:C-gen-T},
replacing $C_r$ with $C_f$, shows that 
\begin{equation}\begin{split}
  \label{eq:T-upper-CoF}
  C_{\text{gen}}(T) &\leq \sup_{\rho \text{ on } A\otimes C} C_f\bigl((T\otimes\id)\rho\bigr) - C_f(\rho) \\
                    &=    \sup_k P_f(T\otimes\id_k),
\end{split}\end{equation}
with the coherence of formation power, given by
$P_f(T) := \max_\rho C_f\bigl( T(\rho) \bigr) - C_f(\rho)$.

Despite the fact that $C_f(\rho) \geq C_r(\rho)$, since the upper bound
is given by a difference of two coherence measures, it might be that
for certain channels, the bound (\ref{eq:T-upper-CoF}) is better than
(\ref{eq:T-upper}), and vice versa for others.
\end{remark}

\medskip
Since the supremum over $k$ of the coherence powers of $T\otimes\id_k$
play such an important role in our bounds, we introduce notation for them,
\begin{align}
  \mathbb{P}_r(T) &:= \sup_k P_r(T\otimes\id_k), \\
  \mathbb{P}_f(T) &:= \sup_k P_f(T\otimes\id_k), \\
  \widetilde{\mathbb{P}}_r(T) &:= \sup_k \widetilde{P}_r(T\otimes\id_k),
\end{align}
and call them the \emph{complete coherence powers} with respect to relative
entropy of coherence and coherence of formation, respectively.
For these parameters, and the coherence generation capacity of isometries, 
we note the following additivity formulas.

\begin{proposition}
  For any CPTP maps $T_1:A_1 \longrightarrow B_1$ and $T_2:A_2 \longrightarrow B_2$,
  \begin{align*}
    \mathbb{P}_r(T_1\otimes T_2) &= \mathbb{P}_r(T_1) + \mathbb{P}_r(T_2), \\
    \mathbb{P}_f(T_1\otimes T_2) &= \mathbb{P}_f(T_1) + \mathbb{P}_f(T_2).
  \end{align*}
  Furthermore, for isometries $T_i(\rho) = V_i\rho V_i^\dagger$,
  \[
    \widetilde{\mathbb{P}}_r(T_i) = \widetilde{P}_r(T_i) = C_{\text{gen}}(T_i)
  \]
  and
  \[
    \widetilde{P}_r(T_1\otimes T_2) = \widetilde{P}_r(T_1) + \widetilde{P}_r(T_2).
  \]
  In other words, the coherence generating capacity of isometries is additive.
\end{proposition}

\begin{proof}
For $X \in \{r,f\}$, any ancilla system $C$ and any state $\rho$,
\[\begin{split}
  C_X\bigl( (&T_1\otimes T_2\otimes \id_C)\rho \bigr) - C_X(\rho) \\
       &= C_X\bigl( (T_1\otimes T_2\otimes \id_C)\rho \bigr) 
           - C_X\bigl( (T_1\otimes\id_{B_2}\otimes \id_C)\rho \bigr) \\
       &\phantom{==}
          + C_X\bigl( (T_1\otimes\id_{B_2}\otimes \id_C)\rho \bigr) - C_X(\rho) \\
       &= C_X\bigl( (T_2\otimes\id_{A_1C})\sigma \bigr) - C_X(\sigma) \\
       &\phantom{==}
          + C_X\bigl( (T_1\otimes\id_{B_2C})\rho \bigr) - C_X(\rho) \\
       &\leq \mathbb{P}_X(T_2) + \mathbb{P}_X(T_1),
\end{split}\]
where we have introduced the state $\sigma = (T_1\otimes\id_{B_2}\otimes\id_C)\rho$.
By taking the supremum of the left hand side over all ancillas $C$ and
all states $\rho$, we obtain 
$\mathbb{P}_X(T_1\otimes T_2) \leq \mathbb{P}_X(T_1) + \mathbb{P}_X(T_2)$;
since the opposite inequality is trivial, using tensor product ancillas
and tensor product input states, and employing the additivity of $C_r$ and
$C_f$ \cite{Winter}, we have proved the equality.

In the case of isometries, Eq.~(\ref{eq:iso-formula}) in Theorem~\ref{thm:C-gen-T} 
already shows $\widetilde{\mathbb{P}}_r(T_i) = \widetilde{P}_r(T_i) = C_{\text{gen}}(T_i)$.
For the tensor product, we again have trivially
$\widetilde{P}_r(T_1\otimes T_2) \geq \widetilde{P}_r(T_1) + \widetilde{P}_r(T_2)$,
by using tensor product input states.
To get the opposite inequality, we proceed as above: for any pure state 
$\ket{\varphi} \in A_1\otimes A_2$,
\[\begin{split}
  C\bigl( (T_1\otimes T_2)\varphi &\bigr) - C(\varphi) \\
       &= C\bigl( (T_1\otimes T_2)\varphi \bigr) 
           - C\bigl( (T_1\otimes\id_{B_2})\varphi \bigr) \\
       &\phantom{==}
          + C\bigl( (T_1\otimes\id_{B_2})\varphi \bigr) - C(\varphi) \\
       &= C\bigl( (\id_{A_1}\otimes T_2)\psi \bigr) - C(\psi) \\
       &\phantom{==}
          + C\bigl( (T_1\otimes\id_{B_2})\varphi \bigr) - C(\varphi) \\
       &\leq \widetilde{\mathbb{P}}_r(T_2) + \widetilde{\mathbb{P}}_r(T_1) \\
       &=    \widetilde{P}_r(T_2) + \widetilde{P}_r(T_1),
\end{split}\]
with the (pure) state $\psi = (T_1\otimes\id_{B_2})\varphi$, and we are done.
\end{proof}

\section{Implementation of channels: \protect\\ Coherence cost of simulation}
\label{sec:general}
We have seen that a CPTP map can be a resource for coherence
because one can use it to generate coherence from scratch, in the form
of maximally coherent qubit states. True to the
resource paradigm, we have to ask immediately the opposite question:
Is it possible to create the resource using pure coherent states
and only incoherent operations? Here we show that the answer is
generally yes, and we define the asymptotic coherence cost $C_{\text{sim}}(T)$
as the minimum rate of pure state coherence necessary to implement
many independent instances of $T$ using only incoherent operations
otherwise. 

\medskip
We start by recalling the implementation of an arbitrary unitary operation 
$U=\sum_{ij=0}^{d-1} U_{ij} \ketbra{i}{j}$ by means of an 
incoherent operation with Kraus operators 
$\{K_{\alpha}\}$, and using the maximally coherent state 
$\ket{\Psi_{d}}=\frac{1}{\sqrt{d}}\sum_{i=0}^{d-1} \ket{i}$ as resource
\cite[Lemma~2]{Chitambar}. Since it is important for us that the
implementation is in fact by means of a \emph{strictly} incoherent
operation, and for self-containedness, we include the proof.

\begin{proposition}[{Chitambar/Hsieh~\cite{Chitambar}}]
  \label{prop:dxd-unitary}
  Consuming one copy of $\Psi_d$ as a resource,
  any unitary $U=\sum_{i,j=0}^{d-1} U_{ij} \ketbra{i}{j}$ 
  acting on $\CC^d$ can be simulated using strictly incoherent operations.
\end{proposition}
\begin{proof}
We essentially follow~\cite{Chitambar};
see also~\cite{Baumgratz} for the qubit case.
Let the Kraus operators have the following form: 
\[
  K_{\alpha} 
    = \sum_{i,j=0}^{d-1} U_{ij} \ketbra{i}{j} \otimes \ketbra{\alpha}{i\!+\!\alpha\!\!\!\!\!\mod\! d}.
\]
It is easy to see that these Kraus operators satisfy
$K_{\alpha} \Delta K_{\alpha}^{\dagger}\subset \Delta$, 
and that they are indeed strictly incoherent,
for all $\alpha=0,1,\dots,d-1$. Furthermore, with the
notation $q=i+\alpha \!\mod\! d$,
\begin{align*}
\sum_{\alpha=0}^{d-1}K_{\alpha}^{\dagger}K_{\alpha}
  &= \sum_{\alpha=0}^{d-1}\sum_{i,j,k,l=0}^{d-1}U_{ij}U_{kl}^{*} \ketbra{jq}{i\alpha}\ketbra {k\alpha}{lp} \\
  &= \sum_{\alpha=0}^{d-1}\sum_{i,j,k,l=0}^{d-1}U_{ij}U_{kl}^{*} \delta_{i,k}\ketbra{jq}{lp} \\
  &= \sum_{\alpha=0}^{d-1}\sum_{i,j,l=0}^{d-1}U_{ij}U_{il}^{*}\ketbra{jq}{lq} \\
  &= \sum_{\alpha=0}^{d-1}\sum_{j,l=0}^{d-1}\delta_{j,l} \ketbra{jq}{lq} 
   = \1.
\end{align*}
Now, let $\ket{\phi}=\sum_{k=0}^{d-1}\phi_{k}\ket{k}$, then
\[
  K_{\alpha}(\ket{\phi} \otimes \ket{\Psi_d}) = \frac{1}{\sqrt{d}} U \ket{\phi}\otimes\ket{\alpha}.
\]
Thus, under this incoherent operation, and tracing out the ancilla afterwards, 
the system will be in the desired state $U\proj{\phi} U^{\dagger}$ with certainty.
\end{proof}

\medskip
Now we pass to the general case of CPTP maps, which extends
the above result for unitaries, with two different protocols.

\begin{theorem}
  \label{theo:cptp-map}
  Any CPTP map $T : A \longrightarrow B$ can be implemented by incoherent
  operations, using a maximally coherent resource state $\Psi_d$, where $d=|B|$.
\end{theorem}

%\begin{proof}
%\textcolor{blue}{
%Let $T$, CPTP map, be a convex decomposition of extremal CPTP maps, $t_i$, $T=\sum_i p_i t_i$ with $p_i$ the probability. Using Choi's theorem, every extremal CPTP map has a stinespring dilation with an environment of dimension $|E|\leq |A|$, where $A$ is the origin system. \\
%From choi's theorem, every quantum  channel can be written as the sum of $|E|=|A|$ elements of independent Kraus operators
%\[
% \cN (\rho)=\sum_j^{|E|} L_j \rho L_j^\dagger
%\]
%with $L_j=\sum_{ii'}L_{ii'}^{(j)}\ketbra{i}{i'}$. Given Kraus operators of the following form $K_{j\alpha}=\sum_{ii'} L_{ii'}^{(j)}\ketbra{i\alpha}{i'q}$ with $q=i+\alpha [d]$, it's straightforwadly easy to see that
%\[
%  \tr \sum_\alpha K_{j\alpha}(\rho \otimes \psi_d) K_{j\alpha}^\dagger = L_j \rho L_j^\dagger
%\]
%Thus the channel can be implemented using $\psi_d$ maximally coherent state with $d=|A|^2$ and by means of incoherent operations.
%}
%\end{proof}

\begin{proof}
Let $T(\rho)=\sum_{\alpha} K_\alpha \rho K_{\alpha}^\dagger$ be a Kraus decomposition
of $T$, with Kraus operators $K_\alpha : A \longrightarrow B$.
The idea of the simulation is to use teleportation of the output of $T$,
which involves a maximally entangled state $\Phi_D$ on $B^\prime \otimes B^{\prime\prime}$, 
a Bell-measurement on system $B\otimes B^\prime$ with outcomes $jk \in \{0,1,\dots,d-1\}^2$,
and unitaries $U_{jk}$ on $B^{\prime\prime}$. The unitaries $U_{jk}$ can be written as 
$U_{jk}=Z^j X^k$, with the phase and cyclic shift unitaries
\[
  Z=\begin{pmatrix}
      1 &  &  & \\  
      & \omega &  &  \\
      &  & \ddots &  \\
      &  &  &  \omega^{d-1} 
    \end{pmatrix}, \quad
  X=\begin{pmatrix}
      0 &  &  &  1 \\  
      1  & 0 &  &  \\
      & \ddots & \ddots &  \\
      &  & 1 &  0 
  \end{pmatrix},
\]
where $\omega= e^{\frac{2\pi i}{d}}$ is the $d$-th root of unity.
This scheme can be reduced to a destructive (hence incoherent) POVM on 
$A\otimes B^\prime$ with outcomes $jk\alpha$, followed by the application of the incoherent(!)
$U_{jk}$. In detail, the probability of getting outcome $jk$ is 
\[
  \Pr\{jk \mid \sigma\} = \tr \Phi^{(jk)}(K_\alpha \otimes \1)\sigma(K_\alpha \otimes \1)^\dagger,
\]
where $\sigma$ is a state on $A\otimes B^\prime$ and the 
\[
  \ket{\Phi^{(jk)}} = (\1 \otimes Z^j X^k) \ket{\Phi_d}
\]
are the Bell states. We can define the POVM elements 
$M_{jk\alpha} = (K_\alpha \otimes \1)^\dagger \Phi^{(jk)} (K_\alpha \otimes \1)$, so that
\begin{align*}
\tr \bigl((T \otimes \id)\sigma\bigr)\Phi^{(jk)} 
  &= \sum_\alpha \tr (K_\alpha \otimes \1) \sigma (K_\alpha \otimes \1)^\dagger \Phi^{(jk)} \\
  &= \sum_\alpha \tr \sigma (K_\alpha^\dagger \otimes \1) \Phi^{(jk)} (K_\alpha \otimes \1) \\
  &= \tr \left[ \sigma \!\left( \!\sum_\alpha (K_\alpha^\dagger \!\otimes\! \1)\Phi^{(jk)} 
                                                 (K_\alpha \!\otimes\! \1) \!\right) \!\right] \\
  &= \tr \left(\sum_\alpha \sigma M_{jk\alpha} \right) 
   = \tr \sigma M_{jk},
\end{align*}
with $M_{jk}=\sum_\alpha M_{jk\alpha}$. 
This leads to a new equivalent scheme in which, given a state $\rho$ on $A$ and 
a maximally entangled state on $B^\prime \otimes B^{\prime\prime}$, we can apply 
the measurement $M_{jk}$ on $A \otimes B^\prime$ with outcomes $jk$, 
and unitaries $U_{jk}$ acting on $B^{\prime\prime}$. 
Formally, let us define the Kraus operators of the protocol by letting
\begin{equation}
  L_{jk\alpha}:= [\bra{\Phi^{(jk)}}(K_\alpha \otimes \1)]^{AB^\prime} \otimes U_{jk}^{B^{\prime\prime}}.
\end{equation} 
It can be checked readily that they satisfy the normalization condition
\begin{equation*}
\sum_{jk\alpha} L_{jk\alpha}^\dagger L_{jk\alpha}=\sum_{jk\alpha} M_{jk\alpha} \otimes \1 =\1 \otimes \1.
\end{equation*}
Applying the Kraus operators $L_{jk\alpha}$ on all the system we get
\begin{align*}
  L_{jk\alpha}\ket{\phi}^A \ket{\Psi_d}^{B^\prime B^{\prime\prime}}
    &= \bra{\Phi^{(jk)}}(K_\alpha \ket{\phi}) (\1 \otimes U_{jk})\ket{\Psi_d} \\
    &= \bra{\Phi^{(jk)}}K_\alpha \ket{\phi} \ket{\Psi^{(jk)}} \\
    &= \frac{1}{d} K_\alpha \ket{\phi}.
\end{align*}  
Hence, 
$\sum_{jk\alpha} L_{jk\alpha}(\rho \otimes \Phi_d) L_{jk\alpha}^{\dagger}
= \sum_\alpha K_{\alpha} \rho K_{\alpha}^{\dagger}= T(\rho)$,
and the proof is complete.
\end{proof}

\begin{theorem}
  \label{theo:cptp-map-SIO}
  Any CPTP map $T : A \longrightarrow B$ can be implemented by strictly 
  incoherent operations and a maximally coherent state $\Psi_d$, where $d \leq |A||B|$.
\end{theorem}

\begin{proof}
The channel $T$ is, first of all, a convex combination of
extremal CPTP maps $T_\lambda$, each of which has at most $|A|$ Kraus operators
\cite{Choi}: $T=\sum_\lambda p_\lambda T_\lambda$. Clearly, we only have to
prove the claim for the $T_\lambda$. Because of the bound on the Kraus operators,
each $T_\lambda$ has a unitary dilation $U_\lambda: A\otimes B \rightarrow B\otimes A$,
such that $T_\lambda(\rho) = \tr_A U_\lambda(\rho\otimes\proj{0}_B)U_\lambda^\dagger$,
with a fixed incoherent state $\ket{0} \in B$.

As the $U_\lambda$ act on a space of dimension $d=|A||B|$, we can invoke
the simulation according to Proposition~\ref{prop:dxd-unitary}.
\end{proof}

\begin{remark}
Comparing Theorems~\ref{theo:cptp-map} and \ref{theo:cptp-map-SIO},
we note that the latter always consumes more resources, but it is 
guaranteed to be implemented by strictly incoherent operations, 
a much narrower class than the incoherent operations. 

We leave it as an open question whether the resource consumption
of Theorem~\ref{theo:cptp-map} can be achieved by strictly incoherent
operations, or whether there is a performance gap between incoherent
and strictly incoherent operations.

Note that these two classes, incoherent operations and strictly
incoherent operations, are distinct as sets of CPTP maps, although
it is known that they induce the same possible state transformations
of a given state into a target state for qubits \cite{Chitambar-2} 
and for pure states in arbitrary dimension \cite{Zhu} (correcting the
earlier erroneous proof of the claim by \cite{Du1}; the SIO part
of the pure state transformations is due to \cite{Winter}). 
However, for the distillation of pure coherence at rate $C_r(\rho)$ \cite{Winter},
IO are needed and it remains unknown whether SIO can attain the same rate. 
Crucially, any destructive measurement, i.e.~any POVM followed by 
an incoherent state preparation, is IO, but the only measurements
allowed under SIO are of diagonal, i.e.~incoherent, observables.
\end{remark}

\medskip
The results so far are about the resources required for the exact implementation
of a single instance of a channel, in the worst case. It is intuitively
clear that some channels are easier to implement in the sense that fewer
resources are needed; e.g.~for the identity or any incoherent channel,
no coherent resource is required.

In the spirit of the previous section, we are interested in the minimum
resources required to implement many independent instances of $T$.

\begin{definition}
  \label{def:simulation}
  An $n$-block incoherent simulation of a channel $T:A\longrightarrow B$ 
  with error $\epsilon$ and coherent resource $\Psi_d$ (on space $D$) is an
  incoherent operation $\cI: A^n \otimes D \longrightarrow B^n$, such that
  $T'(\rho) := \cI(\rho \otimes \Psi_d)$ satisfies
  \[\begin{split}
    \epsilon &\geq \| T' - T^{\otimes n} \|_\diamond \\
         &= \sup_{\ket{\phi}\in A^n\otimes C} \| (T'\otimes\id_C)\phi - (T^{\otimes n}\otimes\id_C)\phi \|_1.
  \end{split}\]
  Here, $C$ is an arbitrary ancilla system; the error criterion of the simulation
  is known as diamond norm \cite{diamond} or completely bounded trace
  norm \cite{Paulsen}, see also \cite{Watrous}.
  
  The rate of the simulation is $\frac1n \log d$, and the 
  simulation cost of $T$, denoted $C_{\text{sim}}(T)$, is the smallest
  $R$ such that there exist $n$-block incoherent simulations with error
  going to $0$ and rate going to $R$ as $n\rightarrow\infty$.
\end{definition}

The best general bounds we have on the simulation cost are contained
in the following theorem. 

\begin{theorem}
  \label{thm:sim-bounds}
  For any CPTP map $T:A \longrightarrow B$,
  \begin{equation}
     \label{eq:sim-gen-bound}
     C_{\text{gen}}(T) \leq C_{\text{sim}}(T) \leq \log|B|.
  \end{equation}
  Furthermore,
  \begin{equation}
    \label{eq:sim-vs-power}
    C_{\text{sim}}(T) \geq \max\left\{ \sup_k P_r(T\otimes\id_k),
                                       \sup_k P_f(T\otimes\id_k) \right\},
  \end{equation}
  where we recall the definitions of the relative entropy coherence power, 
  $P_r(T) = \max_\rho C_r\bigl( T(\rho) \bigr) - C_r(\rho)$,
  and of the coherence of formation power,
  $P_f(T) = \max_\rho C_f\bigl( T(\rho) \bigr) - C_f(\rho)$.
\end{theorem}

\begin{proof}
We start with Eq.~(\ref{eq:sim-gen-bound}):
The upper bound is a direct consequence of Theorem~\ref{theo:cptp-map}.
The lower bound follows from the fact that $T$ is implemented
using maximally coherent states at rate $R = C_{\text{sim}}(T)$
and incoherent operations. Generation of entanglement on the other hand
uses $T$ and some more incoherent operations. Since incoherent operations
cannot increase the amount of entanglement, the overall process of
simulation and generation cannot result in a rate of coherence of
more than $R$.

Regarding Eq.~(\ref{eq:sim-vs-power}), the idea is that for $\epsilon>0$
and $n$ large enough, since the simulation implements a CPTP map
$T'$ that is within diamond norm $\epsilon$ from $T^{\otimes n}$,
using incoherent operations and $\Psi_2^{\otimes n(R+\epsilon)}$ as
a resource. Applying the simulation to the state $\rho^{\otimes n}$, results in 
$(T'\otimes\id_{k^n})\rho^{\otimes n}\approx \bigl((T\otimes\id_k)\rho\bigr)^{\otimes n}$,
hence we have an overall incoherent operation
\[
  \Psi_2^{\otimes n(R+\epsilon)} \otimes \rho^{\otimes n} 
                          \stackrel{\text{IO}}{\longmapsto} (T'\otimes\id_{k^n})\rho^{\otimes n}.
\]
By monotonicity of $C_X$ ($X\in\{r,f\}$) under IO, and $C_X(\Psi_2)=1$, this means
\[
  n(R+\epsilon) \geq C_X\bigl( (T'\otimes\id_{k^n})\rho^{\otimes n} \bigr) - C_X(\rho^{\otimes n}),
\]
where we have used additivity of $C_r$ and $C_f$ \cite{Winter}. Since this holds
for all $\rho$, we obtain
\[\begin{split}
  n(R+\epsilon) &\geq P_X( T'\otimes\id_{k^n} )                                     \\
                &\geq P_X( T^{\otimes n}\otimes\id_{k^n} ) - n \kappa_X \epsilon - 4 \\
                &\geq n P_X( T\otimes\id_k) - n \kappa_X \epsilon - 2,
\end{split}\]
invoking in the second line Lemma~\ref{lemma:power-continuous} below,
with $\kappa_r = 4\log|B|$ and $\kappa_f = \log|B|+\log k$,
and in the third a tensor power test state. Since $\epsilon$ can be
made arbitrarily small, and $n$ as well as $k$ arbitrarily large, the 
claim follows.
\end{proof}

\medskip
Here we state the technical lemma required in the proof of Theorem~\ref{thm:sim-bounds}.

\begin{lemma}
  \label{lemma:power-continuous}
  The relative entropy coherence power and the coherence of formation power 
  are asymptotically continuous with respect to the diamond norm metric on channels. 
  To be precise, for $T_1,T_2 : A \longrightarrow B$ 
  with $\frac12\| T_1-T_2 \|_\diamond \leq \epsilon$,
  \begin{align}
    \label{eq:conti-Pr}
    \bigl| P_r(T_1\otimes\id_k)-P_r(T_2\otimes\id_k) \bigr| &\leq 4\epsilon\log|B| + 2g(\epsilon), \\
    \label{eq:conti-Pf}
    \bigl| P_f(T_1\otimes\id_k)-P_f(T_2\otimes\id_k) \bigr| &\leq \epsilon(\log|B|\!+\!\log k) + g(\epsilon),
  \end{align}
  where $g(x) = (1+x) h_2\!\left(\!\frac{x}{1+x}\!\right) = (1+x)\log(1+x)-x\log x$.
\end{lemma}

\begin{proof}
For the first bound, observe
\[\begin{split}
  \bigl| P_r&(T_1\otimes\id_k)-P_r(T_2\otimes\id_k) \bigr| \\
            &\leq \max_{\rho^{AC}}\ \bigl| C_r\bigl((T_1\otimes\id_k)\rho\bigr)
                                          - C_r\bigl((T_2\otimes\id_k)\rho\bigr) \bigr| \\
            &=    \max_{\rho^{AC}}\ \bigl| S(BC)_{(\Delta T_1\otimes\Delta)\rho}
                                            - S(BC)_{(T_1\otimes\id_k)\rho} \bigr. \\
            &\phantom{=====}
                                     \bigl. - S(BC)_{(\Delta T_2\otimes\Delta)\rho}
                                            + S(BC)_{(T_2\otimes\id_k)\rho} \bigr| \\
            &=    \max_{\rho^{AC}}\ \bigl| S(B|C)_{(\Delta T_1\otimes\Delta)\rho}
                                            - S(B|C)_{(T_1\otimes\id_k)\rho} \bigr. \\
            &\phantom{=====}
                                     \bigl. - S(B|C)_{(\Delta T_2\otimes\Delta)\rho}
                                            + S(B|C)_{(T_2\otimes\id_k)\rho} \bigr| \\
            &\leq \max_{\rho^{AC}}\ \bigl| S(B|C)_{(\Delta T_1\otimes\Delta)\rho}
                                            - S(B|C)_{(\Delta T_2\otimes\Delta)\rho} \bigr| \\
            &\phantom{=====}
                                     + \bigl| S(B|C)_{(T_1\otimes\id_k)\rho}
                                              - S(B|C)_{(T_2\otimes\id_k)\rho} \bigr| \\
            &\leq 2\bigl( 2\epsilon\log|B|+g(\epsilon) \bigr),
\end{split}\]
where in the first line we insert the same variable $\rho^{AC}$ to
maximise $P_r(T_j\otimes\id_k)$ and notice that in that case,
the term $C_r(\rho)$ cancels; then in the second line, we use the 
definition of the relative entropy of coherence and in the third
we use chain rule $S(BC) = S(B|C)+S(C)$ for the entropy, allowing
us to cancel matching $S(C)$ terms; in the fourth line we invoke
the triangle inequality, and finally the Alicki-Fannes bound for the
conditional entropy \cite{AlickiFannes} in the form given in \cite[Lemma~2]{Winter-S-cont}.

For the second bound, we start very similarly:
\[\begin{split}
  \bigl| P_f&(T_1\otimes\id_k)-P_f(T_2\otimes\id_k) \bigr| \\
            &\leq \max_{\rho^{AC}}\ \bigl| C_f\bigl((T_1\otimes\id_k)\rho\bigr)
                                          - C_f\bigl((T_2\otimes\id_k)\rho\bigr) \bigr| \\
            &\leq \epsilon (\log|B|+\log k) + g(\epsilon),
\end{split}\]
where the last line comes directly from the asymptotic continuity
of the coherence of formation \cite[Lemma~15]{Winter}.
We close the proof expressing our belief that it is possible to prove a
version of Eq.~(\ref{eq:conti-Pf}) where $k$ does not appear on the right hand side.
\end{proof}

\medskip
\begin{remark}
As a consequence, while for a channel $T$ that is close to an 
incoherent operation (in diamond norm), or in fact close to a MIO
operation, the coherence generating capacity $C_{\text{gen}}(T)$
is also close to $0$, we do not know at the moment whether the
same holds for the simulation cost $C_{\text{sim}}(T)$.
\end{remark}

\section{Qubit unitaries}
\label{sec:qubit}
In this section, we want to have a closer look at qubit unitaries, for which
we would like to find the coherence generating capacity and simulation cost.

To start our analysis, we note that a general $2\times 2$-unitary has four 
real parameters, but we can transform
unitaries into each other at no cost by preceding or following them by incoherent
unitaries, i.e.~combinations of the bit flip $\sigma_x$ and diagonal
(phase) unitaries $\begin{pmatrix} e^{i\alpha} & 0 \\ 0 & e^{i\beta} \end{pmatrix}$.
This implies an equivalence relation among qubit unitaries up to
incoherent unitaries. A unique representative of each equivalence class is
given by
\begin{equation}
  U = U(\theta) = \begin{pmatrix} c & -s \\ s & c \end{pmatrix},
\end{equation}
where $c=\cos\theta$ and $s=\sin\theta$ and with $0\leq \theta \leq \frac{\pi}{4}$,
so that $c \geq s \geq 0$.

\medskip
One can calculate $C_{\text{gen}}(U(\theta))$,
using the formula from Theorem~\ref{thm:C-gen-T}. 
\begin{figure}[ht]
  \includegraphics[width=8.5cm]{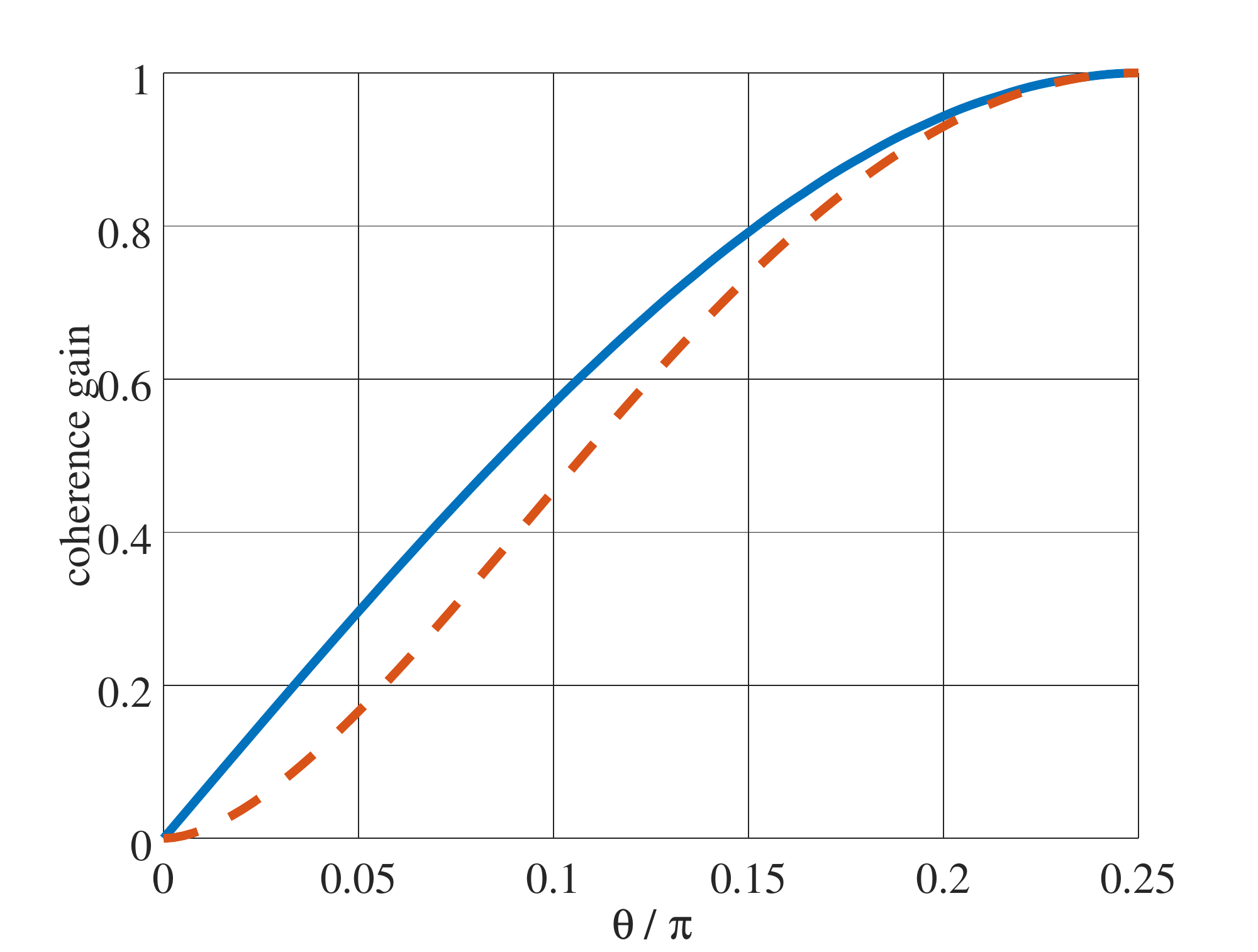}
  \caption{Plot of $C_{\text{gen}}(U(\theta)) = \widetilde{P}_r(U(\theta))$ as a function
           of $\theta \in [0,\frac{\pi}{4}]$ (solid blue line), 
           and comparison with $h_2(\cos^2\theta)$ (dashed red line),
           which is the coherence generated by an incoherent input state.
           In particular, for $\theta \approx 0$, the ratio between the two functions is unbounded.
           The angle $\theta$ is plotted as a fraction of $\pi$.}
  \label{fig:U-theta}
\end{figure}
Clearly, by choosing the test state $\varphi$ to be pure incoherent,
\begin{equation}\begin{split}
  C_{\text{gen}}(U(\theta)) &=    \widetilde{P}_r(U(\theta)) \\
                            &\geq h_2(c^2)
                             =    -c^2 \log c^2 - s^2 \log s^2,
\end{split}\end{equation}
with $h_2(x) = -x\log x-(1-x)\log(1-x)$ the binary entropy.
Perhaps surprisingly, however, this is in general not the optimal state
\cite[Cor.~5]{Maria:MSc} (see also \cite{Garcia-Diaz}), 
meaning that $\widetilde{P}_r(U(\theta))$ is attained at a coherent 
test state $\varphi$, although no closed form expression seems to be known.
In fact, simple manipulations show that we only need to optimise 
$C\!\left(U(\theta)\varphi U(\theta)^\dagger\right) - C(\varphi)$
over states $\ket{\varphi} = U(\alpha)\ket{0} = \cos\alpha\ket{0} + \sin\alpha\ket{1}$,
$0\leq\alpha\leq \pi$ (i.e.~no phases are necessary). The function to optimise
becomes $h_2\bigl( \cos^2(\alpha+\theta) \bigr) - h_2\bigl( \cos^2 \alpha \bigr)$. 
Its critical points satisfy the transcendental equation
\begin{equation}
  \sin(2\alpha+2\theta)\ln\tan^2(\alpha+\theta) = \sin(2\alpha)\ln\tan^2\alpha,
\end{equation}
which can be solved numerically. Fig.~\ref{fig:U-theta} shows that 
$C_{\text{gen}}(U(\theta)) = \widetilde{P}_r(U(\theta)) > h_2(\cos^2 \theta)$
for across the whole interval, except at the endpoints $\theta = 0,\, \frac{\pi}{4}$; 
in Fig.~\ref{fig:optimal-alpha} we plot the optimal $\alpha$ for $U(\theta)$.
\begin{figure}[ht]
  \includegraphics[width=8.6cm]{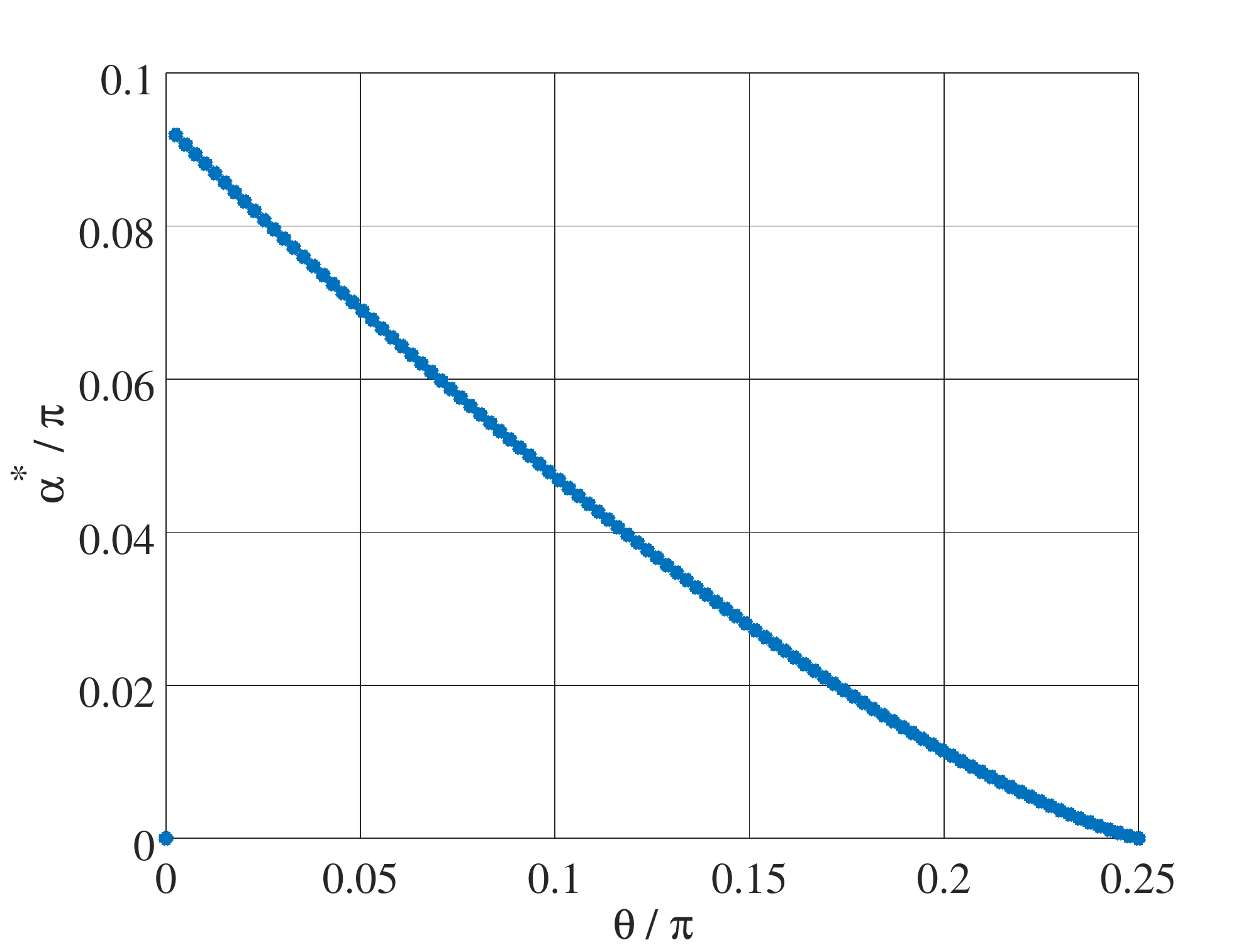}
  \caption{The optimal value of $\alpha$ attaining
           $C_{\text{gen}}(U(\theta)) = \widetilde{P}_r(U(\theta)) 
            = h_2\bigl( \cos^2(\alpha+\theta) \bigr) - h_2\bigl( \cos^2 \alpha \bigr)$,
           as a function of $\theta \in [0,\frac{\pi}{4}]$.
           It is nonzero except at the endpoints $\theta = 0,\,\frac{\pi}{4}$.
           Both angles, $\theta$ and $\alpha^*$, are plotted as fractions of $\pi$.}
  \label{fig:optimal-alpha}
\end{figure}

On the other hand, regarding the implementation of these unitary channels, 
all we can say for the moment is that $C_{\text{sim}}(U(\theta)) \leq 1$, 
because we can implement each instance of the qubit unitary using a qubit 
maximally coherent state $\Psi_2$. 
It is perhaps natural to expect that one could get away with a smaller
amount of coherence, but it turns out that with a two-dimensional 
resource state this is impossible.

\begin{proposition}
  \label{prop:U-theta-sim}
  The only qubit coherent resource state $\ket{r} \in \CC^2$ that permits
  the implementation of $U(\theta)$, $0 < \theta \leq \frac{\pi}{4}$, is the
  maximally coherent state.
  
  Furthermore, any two-qubit incoherent operation $\mathcal{I}$ such that
  $\mathcal{I}(\rho\otimes\proj{r}) = U(\theta)\rho U(\theta)^\dagger \otimes \sigma$
  for general $\rho$, is such that the state $\sigma$ left behind in the ancilla is
  necessarily incoherent.
\end{proposition}

\begin{proof}
We want to know for which state $\ket{r}=c'\ket{0}+s'\ket{1}$ the transformation
$\ket{\psi}\ket{r}\xrightarrow{IO}(U(\theta)\ket{\psi})\ket{0}$ is possible, 
for a general state $\ket{\psi}$. 
Without loss of generality, the incoherent Kraus operators achieving 
the transformation have the following general form:
\begin{equation}
  \label{eq:K}
  K = \lambda \bigl( U(\theta) \otimes \ket{0}\!\bra{r}+R\otimes \ket{0}\!\bra{r^\perp} \bigr),
\end{equation} 
where $\ket{r^\perp}=s'\ket{0}-c'\ket{1}$ is the vector orthogonal to $\ket{r}$. 
We now need to find the form of $R$ such that $K$ is incoherent. 
For that, we impose incoherence of $K$ when tracing out the ancillary part: 
$\bra{0}^AK\ket{0}^A =: T_0$ and $\bra{1}^A K \ket{1}^A =: T_1$, where 
$T_0$ and $T_1$ must be $2$-dimensional incoherent operators. We then obtain
that $R=s'T_0-c'T_1$ and $\lambda U=c'T_0+s'T_1$. The latter condition
enforces that either 
$T_0 \propto \begin{pmatrix}
1 & 0 \\ 
0 & 1
\end{pmatrix}$ and 
$T_1 \propto \begin{pmatrix}
0 & -1 \\ 
1 & 0
\end{pmatrix}$ or viceversa; or 
$T_0 \propto \begin{pmatrix}
1 & -1 \\ 
0 & 0
\end{pmatrix}$ and 
$T_1 \propto \begin{pmatrix}
0 & 0 \\ 
1 & 1
\end{pmatrix}$ or viceversa. 
From these possibilities, we get $4$ possible $R$ matrices, which define $4$ 
different Kraus operators $K_i$ defined, according to Eq.~(\ref{eq:K}),
by $R_i$ matrices as follows:
\begin{align*}
  R_1 &= \begin{pmatrix}
           c\frac{s'}{c'} & -s\frac{c'}{s'} \\ 
           s\frac{c'}{s'} & c\frac{s'}{c'}
          \end{pmatrix}, \\
  R_2 &= \begin{pmatrix}
           -c\frac{c'}{s'} & s\frac{s'}{c'} \\ 
           -s\frac{s'}{c'} & -c\frac{c'}{s'}
         \end{pmatrix}, \\
  R_3 &= \begin{pmatrix}
           -c\frac{c'}{s'} & -s\frac{c'}{s'} \\ 
           -s\frac{s'}{c'} & c\frac{s'}{c'}
         \end{pmatrix}, \\
  R_4 &= \begin{pmatrix}
           c\frac{s'}{c'} & s\frac{s'}{c'} \\ 
           s\frac{c'}{s'} & -c\frac{c'}{s'}
         \end{pmatrix},
\end{align*}
and the general incoherent Kraus operator is $K = \lambda_i K_i$ ($i=1,2,3,4$).
Finally, after imposing $\sum_i|\lambda_i|^2K^\dagger_i K_i=\1$, 
we obtain the following conditions on $R_i$ and $\lambda_i$:
\begin{align*}
  \sum_i|\lambda_i|^2 &= 1, \\
  \sum_i|\lambda_i|^2 R_i^\dagger R_i &= \1, \\
  \sum_i|\lambda_i|^2 R_i &=0.
\end{align*}
It can be verified that these conditions are only fulfilled when $|c'|=|s'|=\frac{1}{\sqrt{2}}$, 
i.e.~$\ket{r}$ is maximally coherent.

If the incoherent implementation of the unitary, instead of mapping
two qubits (input and resource state) to one (output), but to two (output
plus residual resource), 
i.e.~$\mathcal{I}(\rho\otimes\proj{r}) = U(\theta)\rho U(\theta)^\dagger \otimes \sigma$,
then first of all $\sigma$ has to be the same irrespective of the state $\rho$.
Otherwise we would be able, by measuring $\sigma$, to learn some information about
$\rho$ without disturbing it. 
Now consider a pure incoherent input state $\rho = \proj{0}$, and
note that the desired output state $U(\theta)\ket{\psi}$ has nontrivial coherence.
But now observe that $\mathcal{I}$ takes in a state of coherence rank $2$ \cite{Winter},
and produces a product of a pure state of coherence rank $2$ with another
state. Since the coherence rank cannot increase, even under individual 
Kraus operators \cite{Winter}, it must be the case that $\sigma$ is incoherent.
\end{proof}

\medskip
This result might suggest an
irreversibility between simulation and coherence generation for these
unitaries, but we point out that it does not preclude the possibility of 
simulations using a higher rank, yet less coherent, resource state 
(cf.~\cite{stahlke}, where the analogue is demonstrated for LOCC 
implementation of bipartite unitaries using entangled resources); 
or of a simulation of many instances of $U(\theta)$ at a cost lower 
than $1$ per unitary.

\section{Conclusion and Outlook}
\label{sec:outlook}
We have shown that using a maximally coherent state and strictly incoherent 
operations, we can implement any unitary on a system (using $\Psi$ 
for a $d$-dimensional unitary), and via the Stinespring dilation any 
CPTP map (using $\Psi_{d^2}$ for a channel acting on a $d$-dimensional system). 
By teleportation, we prove that for incoherent operations and a $d$-dimensional  
maximally coherent state any noisy channel can be implemented.

Vice versa, every incoherent operation gives rise to some capacity of
generating pure coherence by using it asymptotically many times, and
we have given capacity bounds in general, and a single-letter formula
for the case of unitaries. We also found the additive upper bounds
$\mathbb{P}_r(T)$ and $\mathbb{P}_f(T)$ on the
coherence generation capacity $C_{\text{gen}}(T)$, even though we do
not know whether these numbers are efficiently computable due to the
presence of the extension $\otimes\id_k$, nor whether these extensions
are even necessary.
It is open at the moment whether the
coherence generation capacity $C_{\text{gen}}(T)$ itself is additive for
general tensor product channels, and likewise the lower bound
given in Theorem~\ref{thm:C-gen-T}, 
$\widetilde{\mathbb{P}}_r(T) = \sup_k \widetilde{P}_r(T\otimes\id_k)$;
at least for isometric channels they are.

The coherence generating capacity is never larger than the simulation 
cost, but in general these two numbers will be different. As an extreme
case, consider any CPTP map $T$ that is not incoherent, but is a so-called
``maximally incoherent operation'' (MIO, cf.~\cite{Streltsov-rev}),
meaning that $T(\rho)\in\Delta$ for all $\rho\in\Delta$. This
class was considered in \cite{Brandao3}, and it makes coherence theory
asymptotically reversible \cite{Winter}, all states $\rho$ being 
equivalent to $C_r(\rho)$ maximally coherent qubit states. Such maps exist 
even in qubits, cf.~\cite{Chitambar-erratum} correcting \cite[Thm.~21]{Chitambar-2}.
We expect the simulation cost of any such $T$ to be positive,
$C_{\text{sim}}(T) > 0$.
At the same time, $C_{\text{gen}}(T)=0$ by Theorem~\ref{thm:C-gen-T},
because the relative entropy of coherence is a MIO monotone, and the
tensor product of MIO transformations is MIO.
To obtain an example, we can take any MIO channel for which there exists
a state $\rho$ such that $C_f\bigl(T(\rho)\bigr) > C_f(\rho)$, since
by Theorem \ref{thm:sim-bounds}, $C_{\text{sim}}(T)$ is lower bounded 
by the difference of the two. 
(As an aside, we note that this cannot be realised in qubits,
because for qubits, any state transformation possible under MIO is 
already possible under IO, for which $C_f$ is a monotone \cite{Chitambar-2,Hu}.)
Concretely, consider the following states on a $2d$-dimensional
system $A$, which could be called \emph{coherent flower states}, since
their corresponding maximally correlated states (cf.~\cite{Adesso-corr,Winter})
are the well-known flower states \cite{HHHO}. 
We write them as $2\times 2$-block matrices,
\begin{equation}
  \rho_d = \frac{1}{2d} \begin{pmatrix} \1 & U \\ U^\dagger & \1 \end{pmatrix},
\end{equation}
where $U$ is the $d$-dimensional discrete Fourier transform matrix.
Via the correspondence between $C_r$ and the relative entropy of
entanglement, and between $C_f$ and the entanglement of formation,
respectively, of the associated 
 state, we know that $C_r(\rho_d) = 1$
and $C_f(\rho_d) = 1+\frac12\log d$ \cite{HHHO}. By the results of
\cite{Brandao3}, however, for every $\epsilon>0$ and sufficiently large
$n$, there exists a MIO transformation
$T^{(n)}:(\CC^2)^{\otimes n(1+\epsilon)} \longrightarrow A^n$
with
\begin{equation}
  \rho^{(n)} := T^{(n)}\!\left(\Psi_2^{\otimes n(1+\epsilon)}\right) 
                              \stackrel{\epsilon}{\approx} \rho_d^{\otimes n}.
\end{equation}
%Since $T^{(n)}$ is MIO, and the tensor product of MIO transformations
%is MIO, too, Theorem~\ref{thm:C-gen-T}, Eq.~(\ref{eq:T-upper}) shows that 
%$C_{\text{gen}}(T^{(n)})=0$.
By the asymptotic continuity of $C_f$ \cite{Winter}, we have
$C_f(\rho^{(n)}) \geq n\left(1+\frac12\log d\right) - n \epsilon \log(2d) - g(\epsilon)$,
while of course the preimage $\rho_0 = \Psi_2^{\otimes n(1+\epsilon)}$ has 
$C_f(\rho_0) \leq n + n\epsilon$, so for $\epsilon$ small enough and $n$ large enough,
we have a gap:
\begin{equation}
  C_{\text{sim}}\bigl(T^{(n)}\bigr) \geq \frac{n}{2}\log d - n\epsilon(2+\log d) - g(\epsilon) > 0,
\end{equation}
invoking Theorem~\ref{thm:sim-bounds}.
Thus, while \emph{states} in the resource theory of coherence cannot
exhibit bound resource --- indeed, it was observed in \cite{Winter} that
vanishing distillable coherence, $C_r(\rho)=0$, implies vanishing
coherence cost $C_f(\rho)=0$ ---, operations can have bound coherence,
and the gap can be large on the scale of the logarithm of the channel dimension.
We observe that this effect can only occur for maximally incoherent operations,
which are precisely the ones with $C_{\text{gen}}(T)=0$. We argued already
that MIO channels have zero coherence generating capacity; in the other
direction, if $T$ is not MIO, it means that there exists an incoherent
state $\rho$ such that $T(\rho)$ has coherence, and this can be
distilled at rate $C_r(T(\rho))$ \cite{Winter}.
As MIO are closed under forming tensor products, it also follows
that $C_{\text{gen}}$ does not exhibit superactivation.

This example and the subsequent considerations
raise the question of how our theory 
would change if we considered all MIO transformations as free operations. 
By definition, the above example -- by virtue of being MIO -- has zero
MIO-simulation cost, so there is no bound coherence any more. It may still
be the case that there is in general a difference between MIO-simulation cost
and MIO-coherence generating capacity, but deciding this possibility is
beyond the scope of the present investigation. 
We only note that Theorem~\ref{thm:C-gen-T} gives us a single-letter 
formula for the MIO-coherence generating capacity, namely 
\begin{equation}
  C_{\text{gen}}^{\text{MIO}}(T) = \mathbb{P}_r(T) 
             = \sup_{\rho \text{ on } A\otimes C} C_r\bigl((T\otimes\id)\rho\bigr) - C_r(\rho),
\end{equation}
the complete relative entropy coherence power of $T$. The supremum is
over all auxiliary systems $C$ and mixed states $\rho$ on $A\otimes C$.
Indeed, the upper bound of Eq.~(\ref{eq:T-upper}) still applies, because
$C_r$ is MIO monotone, which is all we needed in the proof of Theorem~\ref{thm:C-gen-T}.
For the lower bound, that it is attainable follows the same idea 
as the proof of Eq.~(\ref{eq:T-lower}), only that we can now use an 
arbitrary mixed state $\rho$ in the argument, since its coherence cost under 
MIO equals $C_r(\rho)$ \cite{Brandao3,Winter}.

\medskip
One of the most exciting possibilities presented by the point of view
of channels as coherence resources is the transformation of channels
into channels by means of preceding and post-processing a given one by
incoherent operations to obtain a different one. The fundamental
question one can ask here is how efficiently, i.e.~at which rate 
$R(T_1\!\rightarrow\! T_2)$ one can transform asymptotically many
instances of $T_1$ into instances of $T_2$, with asymptotically
vanishing diamond norm error. To make nontrivial statements about
these rates, one would need to extend some of the various coherence monotones
that have been studied for states to CPTP maps.

%For qubit unitaries, we could show a stronger result:
%the simulation of such unitary requires much less coherence, in the
%form of a partially coherent state, to implement. The minimally
%coherent state able to implement the unitary via IC is precisely
%the coherent state created by the unitary when acting on an incoherent
%basis state.

%This then shows how to simulate a unitary operation $V$ from another unitary $U$,
%which is perhaps of more interest than transforming a state to another 
%state, since it describes the evolution of the system. 
%In the case of $2 \times 2$ unitary, the transformation is possible for 
%$0 \leq \psi \leq \theta \leq \frac{\pi}{4}$. In order to implement $V$ 
%from $U$ via IC, a partially coherent state has to be consumed. 
%In the asymptotic limits, such transformation is done with a rate upper
%bounded (the optimal rate) by the relative entropy of coherence. 
%\textcolor{red}{For a larger dimension ($d \times d$ untaries), \ldots}

%\vfill

\acknowledgments
The present project was started during two visits of KBD to the 
Quantum Information Group at the Universitat Aut\`onoma de Barcelona. 
The authors thank Ludovico Lami, Swapan Rana, Alexander Streltsov,
Michalis Skotiniotis, and Mark Wilde for interesting discussions
and comments on the coherence of channels.
MGD is supported by a doctoral studies fellowship of the
Fundaci\'on ``la Caixa''. % (Obra Social).
AW is or was supported by the European Commission (STREP ``RAQUEL'')
and the ERC (Advanced Grant ``IRQUAT'').
The authors acknowledge furthermore funding by
the Spanish MINECO (grants FIS2013-40627-P and FIS2016-86681-P),
with the support of FEDER funds, and by the 
Generalitat de Catalunya CIRIT, project 2014-SGR-966.

%\vspace{0.8cm}

\end{document}